\documentclass[a4paper, 12pt]{article}

\usepackage{amssymb}
\usepackage{mathrsfs}
\usepackage{amsfonts}
\usepackage{latexsym,amsthm,amsmath}

\newtheorem{thm}{Theorem}[section]

\newtheorem{rem}[thm]{Remark}

\title{\bf Landau operator on the quaternionic field}
\date{}
\author {
Azzouz Zinoun  \footnote {Laboratoire Phlam, UMR-CNRS 8523, UFR de
Physique Universit\'e Lille I --- Sciences et Technologies, 59655
Villeneuve d'Ascq Cedex, France
\newline
 E-mail: Azzouz.Zinoun@univ-lille1.fr}
, Dominique Kazmierowski \footnote{ Universit\'e Lille I ---
Sciences et Technologies, 59655 Villeneuve d'Ascq Cedex, France
\newline
 E-mail: dominique.kazmierowski@gmail.com}
\\and Ahmed Intissar \footnote{D\'epartement de Math\'ematiques
 Facult\'{e} des Sciences de Rabat BP 1014, Morocco
 \newline
  E-mail: intissar@fsr.ac.ma }}
\begin{document}
\maketitle
\begin{abstract}
\noindent The Landau operator on the quaternionic field is a
differential operator $\mathcal{H}_ {\vec{B}}$ acting on $C^\infty
\left({\mathbb H;\mathbb C} \right)$, defined as the Fourier
transform of the sub-Laplacian associated to the quaternionic
Heisenberg group $\Im m\mathbb H\times_\omega\mathbb H$, where
$\mathbb H \sim\mathbb{R}^4$ is viewed as the space of quaternions,
$\omega$ is a canonical $\Im m\mathbb H$-valued symplectic form on
$\mathbb H $, $\Im m\mathbb H\sim\mathbb{R}^3$ and $ \vec{B} \in \Im
m\mathbb H$ is fixed. $\mathcal{H}_ {\vec{B}}$ is the Hamiltonian of
a charged particle in $ \mathbb H$, interacting with a uniform
magnetic field $\vec {B}$. By a suitable orthogonal change of basis
in $\mathbb H$, $\mathcal{H}_ {\vec{B}}$ is transformed into an
other Landau  operator $\mathcal{H}_{\| \vec{B}\|}$ which is much
simpler, $\| \vec{B}\|$ is the norm of the magnetic field $\vec{B}$.
This new operator is the Hamiltonian of two superposed uncoupled
complex harmonic oscillators.
\end{abstract}
\section{Introduction}
Landau operator arises in both physics and mathematics, and appears
in the framework of the study of charged particles acted on by
magnetic fields. The aim of this paper is to present from a
mathematical point of view an operator analogue to the well known
quantum mechanical Hamiltonian of a charged particle moving in a
uniform magnetic field $\cite{LL}$; this operator will be denoted by
$\mathcal{H}_{\vec B }$. $\mathcal{H}_{\vec B}$ acts on
$C^\infty(\mathbb {H},\mathbb {C})$, where $\mathbb H$ is the
quaternion space and $\mathbb {C}$ is the complex field and
$\vec{B}$ is a uniform magnetic field. In Section 2, we give a brief
review of the Schr\"{o}dinger equation. In Section 3, we construct
the Heisenberg group and its Heisenberg algebra from which we derive
the Landau operator. In Section 4, we show that the Landau operator
$\mathcal{H}_{\vec B }$ is defined as the Fourier transform of the
sub-Laplacian associated to the quaternionic Heisenberg group of
dimension 7; we give the expressions of the Landau operator both in
real and complex form; we define also the angular momentum operator
on $\mathbb H \sim \mathbb{R}^4$. In section 5, we state a theorem
which allows us to give a canonical expression $\mathcal{H}_{\|
\vec{B}\|}$ of $\mathcal{H}_{\vec B }$, which is more simpler to
study. In Section 6, the symmetry group and the gauge invariance of
the Landau operator are discussed. Section 7 relates this work with
studies by other people on a quantum system in a uniform magnetic
field.
\section{Physical motivation and a brief review of Schr\"{o}dinger equation}
The standard procedure for extending the equations of motion of
classical mechanics to the quantum Schr\"{o}dinger equation (in
position space) is through a generalization of the Hamiltonian
formulation of classical mechanics. For the free particle in one
dimension this consist of simply making the identification
\begin{equation}
H_{classical} = \frac{p^2}{2m} \Longrightarrow
\hat{H}=\frac{\hat{p}^2}{2m}.
\end{equation}
The classical Hamiltonian function appropriate for a charged
particle (of charge $q$ and mass $m$) acted on by external electric
and magnetic fields (in three dimensions) is given by
\begin{equation}
H_{classical} = \frac{1}{2m} (\vec{p}-q\vec{A})^2 + q\phi.
\end{equation}
The corresponding quantum mechanical Hamiltonian is obtained by
replacing the momentum variable by its operator counterpart giving
the Schr\"{o}dinger equation
\begin{equation}
\hat{H}\psi(\vec{r},t)
=E\psi(\vec{r},t)=i\hbar\frac{\partial}{\partial t}\psi(\vec{r},t)
\end{equation}
where
\begin{equation}
\label{E4''} \hat{H} = \frac{1}{2m}
(\hat{\vec{p}}-q\hat{\vec{A}}(\vec{r},t))^2 + q\hat{\phi}(\vec{r},t)
\end{equation}
and $\vec{r}\in \mathbb{R}^3$ is the position and $t\in
\mathbb{R}^+$ is the time.\\ One must be careful of the ordering of
any differential operators, so we find
\begin{eqnarray}
\label{E4'}
[\hat{\vec{p}}-q\hat{\vec{A}}(\vec{r},t)]^2\psi(\vec{r},t)=
-\hbar^2\vec{\nabla}^2\psi(\vec{r},t)+iq\hbar\vec{\nabla}.[\hat{\vec{A}}(\vec{r},t)\psi(\vec{r},t)]\nonumber\\
+iq\hbar\hat{\vec{A}}(\vec{r},t).[\vec{\nabla}\psi(\vec{r},t)]+q^2[\hat{\vec{A}}(\vec{r},t).\hat{\vec{A}}(\vec{r},t)]\psi(\vec{r},t).
\end{eqnarray}
When the magnetic field $\vec{B}$ is uniform we can take the
potential vector on the following form
$\widehat{\vec{A}}=\frac{1}{2}(\widehat{\vec{B}}\wedge\widehat{\vec{r}}\,)$,
where $\widehat{\vec{r}}$ is the position vector operator. For the
sake of simplicity, we will forget the `hat' symbol on the
operators. Suppose the Euclidian oriented space $\mathbb R^3$ is
endowed with the canonical basis $(e_1',e_2',e_3')$ and coordinates
$(x_1,\,x_2,\,x_3)$, the components of the momentum vector operator
are $\displaystyle{(-i\frac{\partial}{\partial
x_1},-i\frac{\partial}{\partial x_2},-i\frac{\partial}{\partial
x_3})}$. Let $(B_1,B_2,B_3)$ denote the components of the magnetic
field, then the potential vector has components $(A_1,A_2,A_3)$
given by
\begin{eqnarray*}
A_1&=&\frac{1}{2}(B_2 x_3-B_3x_2)\\
A_2&=&\frac{1}{2}(B_3x_1-B_1x_3)\\
A_3&=&\frac{1}{2}(B_1x_2-B_2x_1).
\end{eqnarray*}
We define the differential 1-form $A=A_1\,dx_1+A_2\,dx_2+A_3\,dx_3$.
If we identify at each point $\left({x_1,x_2,x_3} \right)$ the basis
$\left({dx_2\wedge dx_3,dx_3\wedge dx_1,dx_1\wedge dx_2}\right)$ of
the 2-forms with the canonical basis $\left({e'_1,e'_2,e'_3}\right)$
of $\mathbb R ^3$, then we get $dA=\vec B$ since the magnetic field
is uniform.
The corresponding Hamiltonian operator $H$ at
$x=\left({x_1,x_2,x_3}\right)$ is then equal to
 \begin{eqnarray}
H&=&-\frac{\hbar^2}{2m}\bigg\{ {\left( {\frac{\partial } {{\partial
x_1 }}}-i\frac{q}{\hbar}A_1 \right)^{\!2} + \left( {\frac{\partial }
{{\partial x_2 }} -i\frac{q}{\hbar}A_2} \right)^{\!2} + \left(
{\frac{\partial }
{{\partial x_3 }} -i\frac{q}{\hbar}A_3 } \right)^{\!2} }\bigg\}   \nonumber \\
&=&-\frac{\hbar^2}{2m}\bigg\{\frac{{\partial ^2 }} {{\partial x_1^2
}} + \frac{{\partial ^2 }} {{\partial x_2^2 }} + \frac{{\partial ^2
}} {{\partial x_3^2 }} - i\frac{q}{\hbar}B_1 \left( {x_2
\frac{\partial } {{\partial x_3 }} - x_3 \frac{\partial } {{\partial
x_2}}}\right) \nonumber\\
&&- i\frac{q}{\hbar}B_2 \left( {x_3 \frac{\partial } {{\partial x_1
}} - x_1 \frac{\partial } {{\partial x_3}}}\right)-
i\frac{q}{\hbar}B_3\left( {x_1 \frac{\partial } {{\partial x_2 }} -
x_2 \frac{\partial } {{\partial x_1}}} \right) \nonumber\\
&&- \frac{q^2}{4\hbar^2}\|\vec{B}
\|^2\|\vec{x}\|^2+\frac{q^2}{4\hbar^2}(\vec{B}.\vec{x})^2\bigg\}.
\end{eqnarray}
The orbital momentum is :
\begin{eqnarray*}
\vec{L}=\vec{r}\wedge\vec{p}=\left(
\begin{array}{c}L_1\\L_2\\L_3\end{array}\right)=-i\hbar\left(
\begin{array}{c}\displaystyle{{x_2
\frac{\partial } {{\partial x_3 }} - x_3 \frac{\partial } {{\partial
x_2}}}}\\\displaystyle{{x_3 \frac{\partial } {{\partial x_1 }} - x_1
\frac{\partial } {{\partial x_3}}}}\\\displaystyle{{x_1
\frac{\partial } {{\partial x_2 }} - x_2 \frac{\partial } {{\partial
x_1}}}}\end{array}\right).
\end{eqnarray*}
The Hamiltonian $H$ can be written in the following form :
\begin{eqnarray*}
H&=&-\frac{\hbar^2}{2m}\bigg\{\triangle +
\frac{q}{\hbar^2}\vec{L}.\vec{B}- \frac{q^2}{4\hbar^2}\|\vec{B}
\|^2\|\vec{x}\|^2+\frac{q^2}{4\hbar^2}(\vec{B}.\vec{x})^2\bigg\}
\hfill \nonumber
\end{eqnarray*}
If $\vec{B}$ is directed along $e'_3$ then :
\begin{eqnarray*}
H&=&-\frac{\hbar^2}{2m}\bigg\{\frac{{\partial ^2 }} {{\partial x_1^2
}} + \frac{{\partial ^2 }} {{\partial x_2^2 }} + \frac{{\partial ^2
}} {{\partial x_3^2 }} - i\frac{q}{\hbar}B_3 \left( {x_1
\frac{\partial } {{\partial x_2 }} - x_2 \frac{\partial } {{\partial
x_1}}} \right) \\
&&- \frac{q^2}{4\hbar^2}{B_3} ^2(x_1^2 +x_2^2)\bigg\}\\
&=& H_\bot +H_{//}.
\end{eqnarray*}
 In the $\left( {x_1 ,x_2 } \right)$-plane, the particle is
described by the operator
\[ H_\bot = -\frac{\hbar^2}{2m}\bigg\{\frac{{\partial ^2 }} {{\partial x_1^2
}} + \frac{{\partial ^2 }} {{\partial x_2^2 }} - i\frac{q}{\hbar}B_3
\left( {x_1 \frac{\partial } {{\partial x_2 }} - x_2 \frac{\partial
} {{\partial x_1}}} \right) - \frac{q^2}{4\hbar^2}{B_3} ^2(x_1^2
+x_2^2)\bigg\}. \hfill \nonumber\]
$H_{\bot}$ is called the \emph{Landau operator}.\\
The operator $\displaystyle
 H_{//}=-\frac{\hbar^2}{2m}\frac{{\partial ^2}}{\partial x_3^2} $
is the Hamiltonian along the $x_3$-axis.

If we introduce the complex variable $z=x_1+ix_2$, the Landau
operator becomes
\[H_{\bot} =-\frac{\hbar^2}{2m}\bigg\{ 4\frac{\partial }
{{\partial z}}\frac{\partial } {{\partial \bar z}} +
\frac{q}{\hbar}B_3  \left( {z\frac{\partial } {{\partial z}} - \bar
z\frac{\partial } {{\partial \bar z}}} \right)
-\frac{q^2}{4\hbar^2}{B_3} ^2  \left| z \right|^2\bigg\}. \]
All these expressions of $H$ represent the well known Hamiltonian of
a charged particle in a uniform magnetic field (see $\cite{LL}$ and
$\cite{GI}$).

\section{\large Mathematical point of view and the statement of the main
results.}
The space we work in is the Euclidean space $\mathbb {R}^4$
identified with the quaternionic space $\mathbb{H}$. In the basis
$(e_0,e_1,e_2,e_3)$, an element $x$ of $\mathbb R^4$ is written in
the following form :
\begin{eqnarray}
\label{E5}
x&=&x_0e_0+x_1e_1+x_2e_2+x_3e_3 \\
&=&(x_0,\vec{x})
\end{eqnarray}
where $\vec{x}=(x_1,x_2,x_3)\in\mathbb{R}^3\sim\Im m\mathbb{H}$.

 Let $\{1,\textbf{i},\textbf{j},\textbf{k}\}$ be the basis of $\mathbb{H}$, with
\[\textbf{i}^2 = \textbf{j}^2= \textbf{k}^2=-1 \]\[ \textbf{i}\textbf{j}=-\textbf{j}\textbf{i},~~
 \textbf{i}\textbf{k} =-\textbf{k}\textbf{i},~~ \textbf{j}\textbf{k}=-\textbf{k}\textbf{j}.\]
 An element $x$ of $\mathbb{H}$ is written as :
\begin{eqnarray}
\label{E6}
x&=&x_0+x_1\textbf{i}+x_2\textbf{j}+x_3\textbf{k}\\
&=&(x_0,\Im m{x}).
\end{eqnarray}
A $U(1)$ potential will be given by functions $A_{\mu}(x)\in
\mathbb{R}$. It will be convenient if we go further to work with the
differential 1-form on $\mathbb R^4$ defined by
\begin{equation}
A(x)=\sum\limits_{\alpha=0}^3A_\alpha(x)dx_\alpha.
\end{equation}
 The 1-form $A$ has a geometric
 significance: it is the connection form which is used to define
 covariant derivatives.
We shall also write the curvature $F$ as an exterior 2-form
\begin{equation}
F=\frac{1}{2}F_{\mu\nu}dx_\mu\wedge dx_\nu.
\end{equation}
The 2-form $F$ is computed from $A$ by
\begin{eqnarray}
\label{E7}
F &=& d A =\sum\limits_{\mu=0}^3 dA_\mu\wedge dx_\mu \nonumber \\
&=&\frac{1}{2}\sum\limits_{\mu,\nu}F_{\mu\nu}dx_\mu\wedge dx_\nu \nonumber\\
&=&\frac{1}{2}\left(\frac{\partial A_1}{\partial x_0}-\frac{\partial
A_0}{\partial x_1}\right) dx_0\wedge
dx_1+\frac{1}{2}\left(\frac{\partial A_2}{\partial
x_0}-\frac{\partial A_0}{\partial x_2}\right)
dx_0\wedge dx_2+\nonumber\\
&&\frac{1}{2}\left(\frac{\partial A_3}{\partial x_0}-\frac{\partial
A_0}{\partial x_3}\right) dx_0\wedge dx_3
+\frac{1}{2}\left(\frac{\partial A_2}{\partial x_1}-\frac{\partial
A_1}{\partial x_2}\right)
dx_1\wedge dx_2+\nonumber\\
&&\frac{1}{2}\left(\frac{\partial A_3}{\partial x_2}-\frac{\partial
A_2}{\partial x_3}\right) dx_2\wedge dx_3
+\frac{1}{2}\left(\frac{\partial A_1}{\partial x_3}-\frac{\partial
A_3}{\partial x_1}\right) dx_3\wedge dx_1.
\end{eqnarray}
In this work we are concerned by a uniform magnetic field
$\vec{B}\in\mathbb R^3.$ Since $\mathbb R^3$ is identified with $\Im
m\mathbb{H}$, we write
\begin{eqnarray*}
B_q&=&B_1 \textbf{i}+B_2 \textbf{j}+B_3 \textbf{k}\\
&=&\vec B
\end{eqnarray*}
($B_q$ represents the quaternion version of the magnetic field
$\vec{B}\in\mathbb R^3$).
A matrix representation of $\textbf{i} ,\textbf{j}, \textbf{k}$ on
$\mathbb{R}^4$ is:
\begin{eqnarray}
\label{E i} \textbf{i}=\left(
\begin{array}{cccc}
0&-1&\phantom-0&\phantom-0\\
1&\phantom-0&\phantom-0&\phantom-0\\
0&\phantom-0&\phantom-0&-1\\
0&\phantom-0&\phantom-1&\phantom-0
\end{array}
\right) \\
 \label{E j}\textbf{j}=\left(
\begin{array}{cccc}
0&\phantom-0&-1&\phantom-0\\
0&\phantom-0&\phantom-0&\phantom-1\\
1&\phantom-0&\phantom-0&\phantom-0\\
0&-1&\phantom-0&\phantom-0
\end{array}
\right) \\\label{E k}  \textbf{k}=\left(
\begin{array}{cccc}
0&\phantom-0&\phantom-0&-1\\
0&\phantom-0&-1&\phantom-0\\
0&\phantom-1&\phantom-0&\phantom-0\\
1&\phantom-0&\phantom-0&\phantom-0
\end{array}
\right).
\end{eqnarray}
 The matrix associated to the magnetic field $\vec B$ is :
\begin{eqnarray}
\label{E7"} \Omega_{\vec B}=\left(
\begin{array}{cccc}
0&-B_1&-B_2&-B_3\\
B_1&\phantom-0&-B_3&\phantom-B_2\\
B_2&\phantom-B_3&\phantom-0&-B_1\\
B_3&-B_2&\phantom-B_1&\phantom-0
\end{array}\right).
\end{eqnarray}
Since we deal with uniform fields, the matrix $\Omega_{\vec B}$ can
be obtained from $(\ref{E7})$ by considering that the coefficients
of the 2-form $F$ are constants and that $F$ is self-dual 2-form.
We have
\begin{eqnarray}
 \label{E8} \frac{1}{2}\left(\frac{\partial A_1}{\partial
x_0}-\frac{\partial A_0}{\partial x_1}\right)=
E_1,~\frac{1}{2}\left(\frac{\partial A_2}{\partial
x_0}-\frac{\partial A_0}{\partial x_2}\right)=E_2,~
\frac{1}{2}\left(\frac{\partial A_3}{\partial x_0}-\frac{\partial
A_0}{\partial x_3}\right)=E_3
\end{eqnarray}
\begin{eqnarray}
\label{E9} \frac{1}{2}\left(\frac{\partial A_3}{\partial
x_2}-\frac{\partial A_2}{\partial x_3}\right)=B_1,~
\frac{1}{2}\left(\frac{\partial A_1}{\partial x_3}-\frac{\partial
A_3}{\partial x_1}\right)=B_2,~\frac{1}{2}\left(\frac{\partial
A_2}{\partial x_1}-\frac{\partial A_1}{\partial x_2}\right)=B_3
\end{eqnarray}
The matrix associated to $F$ is :
\begin{eqnarray}
\label{E10''} \Omega_{\vec{E},\vec{B}}=\left(
\begin{array}{cccc}
0&-E_1&-E_2&-E_3\\
E_1&\phantom-0&-B_3&\phantom-B_2\\
E_2&\phantom-B_3&\phantom-0&-B_1\\
E_3&-B_2&\phantom-B_1&\phantom-0
\end{array}\right)
\end{eqnarray}
Note the analogy of the tensor $\Omega_{\vec{E},\vec{B}}$ and the
electromagnetic tensor $(F^{\mu\nu})$, which one can derive from the
Maxwell equations (we work in the Euclidean space $\mathbb{R}^4$
instead of the Minkowski space $\mathbb{R}^{3,1}$). This matrix may
be interpreted as the matrix associated to an electric and a
magnetic uniform fields
\[\vec{E}=\left(
\begin{array}{c}E_1\\E_2\\E_3\end{array}\right),\quad \vec{B}=\left(
\begin{array}{c}B_1\\B_2\\B_3\end{array}\right).\]
We define the connection $A$ associated to these uniform fields by
its potential vector
\[ \vec A=\Omega_{\vec {E},\vec{B}}\vec x \]
where
\begin{eqnarray*} \vec x&=&\left(
\begin{array}{c}x_0\\x_1\\x_2\\x_3\end{array}\right)~\in~\mathbb{R}^4,\\
\Omega_{\vec{E},\vec {B}}\vec
x&=&\left(\begin{array}{c}A_0\\A_1\\A_2\\A_3
\end{array} \right)=\left(\begin{array}{c}
-E_1x_1-E_2x_2-E_3x_3\\
\phantom-B_1x_0-B_3x_2+B_2x_3\\
\phantom-B_2x_0+B_3x_1-B_1x_3\\
\phantom-B_3x_0-B_2x_1+B_1x_2
\end{array} \right)\\
\end{eqnarray*}
Let $\star$ be the Hodge operator; the self-duality of the 2-form
$F$ is:
 $\star F=F $, this implies $\vec{E} =\vec{B}$, hence
the matrix $(\ref {E10''})$ becomes :
\[\Omega_{B_{q}}=\left(
\begin{array}{cccc}
0&-B_1&-B_2&-B_3\\
B_1&\phantom-0&-B_3&\phantom-B_2\\
B_2&\phantom-B_3&\phantom-0&-B_1\\
B_3&-B_2&\phantom-B_1&\phantom-0
\end{array}\right)\] which is exactly the matrix $(\ref{E7"})$.
Then,
\begin{eqnarray}
\Omega_{B_{q}}\vec x=\left(\begin{array}{c}A_0\\A_1\\A_2\\A_3
\end{array} \right)=\left(\begin{array}{c}
-B_1x_1-B_2x_2-B_3x_3\\
\phantom- B_1x_0-B_3x_2+B_2x_3\\
\phantom- B_2x_0+B_3x_1-B_1x_3\\
\phantom- B_3x_0-B_2x_1+B_1x_2
\end{array} \right).
\end{eqnarray}
The Landau operator $H_{B_q}$ is obtained from $(\ref{E4''})$ and
$(\ref{E4'})$ (Quantum Mechanics result), (for simplicity we take $
\hbar=1$, $q=1$, and $2m=1$):
\begin{eqnarray}
\label{E10} \nonumber
 \mathcal
H_{B_q}&=&-\left[\left(\!\frac{\partial}{\partial x_0}+
iA_0\!\!\right)^{\!2} \!\!\!+\!\left(\!\frac{\partial}{\partial
x_1} +i A_1\!\!\right)^{\!2} \right.\\
&&\left.+\left(\!\frac{\partial}{\partial x_2} +i
A_2\!\!\right)^{\!2} \!\!\!+\!\left(\!\frac{\partial}{\partial
x_3} +i A_3\!\!\right)^{\!2}\right]\nonumber \\
&&=-\Delta -2i<\Omega_{B_q} \vec x,\nabla> + \|\Omega_{B_q}\vec
x\|^2
\end{eqnarray}
which is the Hamiltonian of a charged particle in
$\mathbb{H}\sim\mathbb{R}^4$ acted on by a uniform magnetic field
$B_q \in \Im m\mathbb{H} \sim \mathbb{R}^3$.
\subsection{Heisenberg group}
We will derive the Landau operator $(\ref{E10})$ from a Lie group
and its associated Lie algebra approach. This Landau operator on the
quaternionic field may be viewed as the Fourier transform of the
sub-Laplacian associated to the quaternionic Heisenberg group $\Im
m\mathbb H\times_\omega\mathbb H$ of dimension $7$, where $\mathbb H
\sim\mathbb{R}^4$ is viewed as the space of quaternions, $\omega$ is
a canonical $\Im m\mathbb H$-valued symplectic form on $\mathbb H $,
$\Im m\mathbb H\sim\mathbb{R}^3$ and $\nu \in \Im m\mathbb H$ is
fixed. This Heisenberg group is associated to the exact sequence
\[
\label{E11} 0\longrightarrow \Im m \mathbb H\longrightarrow \Im m
\mathbb H\times_\omega\mathbb H\longrightarrow \mathbb
H\longrightarrow0 \] that is, \[ 0\longrightarrow
\mathbb{R}^3\longrightarrow
\mathbb{R}^3\times_\omega\mathbb{R}^4\longrightarrow\mathbb{R}^4\longrightarrow0.
\]
\subsubsection{Canonical form on $\mathbb H$}
Define a 2-form $\omega$ on $\mathbb{H}$ by
\[
\omega\left({x,y}\right)=\frac{1}{2}\left({y\bar{x}-x\bar{y}}\right)
\]
for every $x,y\in\mathbb{H},$ where
$\bar{x}=x_{0}\,1-x_{1}\,\textbf{i}-x_{2}\,\textbf{j}-x_{3}\,\textbf{k}$
is the quaternionic conjugate of $x$.
 It is easily checked that $\omega(x,y)\in\Im m\left(\mathbb{H}\right)$ can be written
\[\omega (x,y)=\omega_{1}(x,y)\textbf{i}+\omega_{2}(x,y)\textbf{j}+\omega_{3}(x,y)\textbf{k}\]
where $\omega_{1},\omega_{2},\omega_{3}$ are real valued 2-forms on
$\mathbb{H}.$
\begin{eqnarray*}
\omega_{1}(x,y)&=&x_{0}y_{1}-x_{1}y_{0}+x_{2}y_{3}-x_{3}y_{2}\\
\omega_{2}(x,y)&=&x_{0}y_{2}-x_{2}y_{0}+x_{3}y_{1}-x_{1}y_{3}\\
\omega_{3}(x,y)&=&x_{0}y_{3}-x_{3}y_{0}+x_{1}y_{2}-x_{2}y_{1}.
\end{eqnarray*}
$\omega_{1},\omega_{2},\omega_{3}$ are symplectic forms on
$\mathbb{R}^4 ,$ so that $\omega$ can be viewed as a
$\mathbb{R}^3$-valued
 symplectic form on $\mathbb{R}^4$ given by
\[
\omega\left({x,y}\right)=\omega_{1}(x,y)e'_{1}+\omega_{2}(x,y)e'_{2}+\omega_{3}(x,y)e'_{3}
\]
where $(e'_{1},e'_{2},e'_{3})$ is the canonical basis of
$\mathbb{R}^3$ corresponding to the basis
$(\textbf{i},\textbf{j},\textbf{k})$ of $\Im m(\mathbb{H}).$
 In the following we consider the
{quaternionic Heisenberg group}
$N_\omega=\mathbb{R}^3\times_\omega\mathbb{R}^4$       associated to
the symplectic form $\omega$
 defined above. Note that $N_\omega$ is topologically
 $\mathbb{R}^3\times\mathbb{R}^4\cong\mathbb{R}^7$.

\subsubsection{Heisenberg group associated to $ \omega $}
 Let (t,x) be an element of $(\mathbb
R^3\times\mathbb R^4)$, where:
\[  t\in \mathbb R^3,\;
t=t_1e'_1+t_2e'_2+t_3e'_3 \]
\[ x\in \mathbb
R^4,\; x=x_0 e_0+x_1 e_1+x_2 e_2+x_3 e_3.\]
We define the
multiplication law $\cdot_\omega$ in $\mathbb R^3\times\mathbb R^4$
by
\[(t,x)\cdot_\omega(t',x')=
\Big(t+t'+\omega(x,x'),x+x'\Big).\]It is easily verified that
$\bigg(\mathbb R^3\times\mathbb R^4,\cdot_\omega\bigg)$ is the
Heisenberg group associated to the form $\omega$. A matrix
representation of the Heisenberg group element $(t,x)$ is :
\[\zeta(t,x)=\left(
\begin{array}{*{20}c}
    1 &  0 &  0 & -x_1 &  x_0 & -x_3 &  x_2 & t_1\\
    0 &  1 &  0 & -x_2 &  x_3 &  x_0 & -x_1 & t_2\\
    0 &  0 &  1 & -x_3 & -x_2 &  x_1 &  x_0 & t_3\\
    0 &  0 &  0 &    1 &   0  &    0 &    0 & x_0\\
    0 &  0 &  0 &    0 &   1  &    0 &    0 & x_1\\
    0 &  0 &  0 &    0 &   0  &    1 &    0 & x_2\\
    0 &  0 &  0 &    0 &   0  &    0 &    1 & x_3\\
    0 &  0 &  0 &    0 &   0  &    0 &    0 &   1\\
 \end{array}
\right) \!\! .\]
 It is easily checked that
\[\zeta(t,x)\zeta(t',y)=\zeta(t+t'+\omega (x,y),x+y)\]
and $\zeta(t,x)$ is an invertible matrix with inverse
$\zeta(-t,-x)$. Hence
$N_\omega=\mathbb{R}^3\times_\omega\mathbb{R}^4$
 may be viewed as a Lie subgroup of the affine group of $\mathbb{R}^7$ :
\[
N_\omega \subset \textit{Aff }(\mathbb{R}^7)\subset
GL(\mathbb{R}^8).
\]

\subsubsection{Lie algebra of the Heisenberg group }
One can determine in the usual manner the infinitesimal generators
which characterize the Lie algebra of the Heisenberg group. We
obtain a Heisenberg algebra $\mathfrak{h}$ of dimension 7 generated
by the following vector fields:
\begin{eqnarray}
\label{E_0} F_0 & = &\frac{\partial}{{\partial x_0}} - x_1
\frac{\partial}{{\partial t_1}} - x_2 \frac{\partial}{{\partial
t_2}} - x_3 \frac{\partial}{{\partial t_3}}  \\
\label{E_1} F_1 & = &\frac{\partial}{{\partial x_1}}+
  x_0 \frac{\partial}{{\partial t_1}}
+ x_3 \frac{\partial}{{\partial t_2}} - x_2
\frac{\partial}{{\partial t_3}}  \\
\label{E_2} F_2 & = &\frac{\partial}{{\partial x_2}} - x_3
\frac{\partial}{{\partial t_1}} + x_0 \frac{\partial}{{\partial
t_2}} + x_1 \frac{\partial}{{\partial t_3}}  \\
\label{E_3} F_3  & = &\frac{\partial}{{\partial x_3}}+
  x_2 \frac{\partial}{{\partial t_1}}
- x_1 \frac{\partial}{{\partial t_2}} + x_0
\frac{\partial}{{\partial t_3}}
\end{eqnarray}
and
\[
 T_1 = \frac{\partial}{{\partial t_1}} \, , \quad  T_2 =
\frac{\partial}{{\partial t_2}} \, , \quad  T_3 =
\frac{\partial}{{\partial t_3}} \, .
\]
  The commutation relations between the generators are:
\[ [T_\lambda ,T_\mu     ] = 0\quad (\lambda ,\mu = 1,2,3),\]
\[ [F_\alpha  ,T_\lambda ] = 0\quad (\alpha = 0,\ldots ,3;\quad \lambda = 1,2,3)\]
and
 \begin{eqnarray*} [F_\alpha ,F_\beta ]& =& 2T_{\gamma}
\end{eqnarray*}
where $(\alpha\beta\gamma)$ is a circular permutation of $(123)$.
The center of the algebra $\mathfrak h$ is the abelian ideal
generated by $T_1$, $T_2$ and $T_3$. This algebra is a 2-step
nilpotent Lie algebra.
\section{The Landau operator $\mathcal{H}_{\vec{\nu}}$}
In the universal envelopping algebra $\mathcal{U}(\mathfrak {h})$ of
$\mathfrak{h}$, we define the following quadratic operator, which is
the Laplace element in $\mathcal{U}(\mathfrak {h})$:
\[\mathcal Q =F_0^2+ F_1^2+F_2^2+F_3^2+ T_1^2+ T_2^2+ T_3^2.\]
Explicitly, $\mathcal Q$ takes the following form:
\begin{eqnarray}
 \mathcal Q&=&\left(\frac{\partial}{{\partial
x_0}}  - x_1 \frac{\partial}{{\partial t_1}}- x_2
\frac{\partial}{{\partial t_2}} - x_3 \frac{\partial}{{\partial
t_3}} +\right)^{\!\!2}+\left(\frac{\partial}{{\partial x_1}}+ x_0
\frac{\partial }{\partial t_1} + x_3 \frac{\partial}{\partial t_2 }
- x_2 \frac{\partial}{\partial t_3 } \right)^{\!\!2}
\nonumber\\
&&+\left(\frac{\partial}{{\partial x_2}} - x_3 \frac{\partial
}{\partial t_1} + x_0 \frac{\partial}{\partial t_2} + x_1
\frac{\partial}{\partial t_3}\right)^{\!\!2}+\left(
\frac{\partial}{{\partial x_3}}+x_2 \frac{\partial }{\partial t_1}-
x_1\frac{\partial}{\partial t_2} + x_0 \frac{\partial}{\partial t_3}
 \right)^{\!\!2} \nonumber \\
&&+\frac{\partial ^2}{\partial t_1^2} \; + \; \frac{\partial ^2}
{\partial t_2^2} \; + \; \frac{\partial ^2}{\partial t_3^2} \;
\end{eqnarray}
and may be written
\[ \mathcal Q =   \Delta_{sub}^{N_\omega} + \Delta ^{\mathbb{R}^3} \,  \]
where
\[ \Delta ^{\mathbb{R}^3 }  = \frac{{\partial ^2 }}
{{\partial t_1^2 }} + \frac{{\partial ^2 }} {{\partial t_2^2 }} +
\frac{{\partial ^2 }} {{\partial t_3^2 }}
\]
is the standard Laplacian on $\mathbb{R}^3$ and
$\Delta_{sub}^{N_\omega}$ is the {sub-Laplacian} associated to the
Heisenberg group $N_\omega=\mathbb{R}^3\times_\omega\;\mathbb{R}^4$.
The partial Fourier transform $\mathcal{F}(\mathcal Q)$ of
$\mathcal{Q}$ on the $t$-variable is obtained by replacing
$\displaystyle{\frac{\partial}{{\partial t_\lambda}} }$ by $i\nu
_\lambda$, $\lambda  = 1,2,3$, (where $\vec\nu=(\nu_1,\nu_2,\nu_3)$
is the dual variable of $\vec{t}=(t_1,t_2,t_3)$); this defines an
operator denoted by $\mathcal{F}(\mathcal Q)$ which has the form:
\begin{eqnarray}
\mathcal{F}(\mathcal{Q})=\mathcal{F}(\Delta_{sub}^{N_\omega})+\mathcal{F}(\Delta
^{\mathbb{R}^3}){~~~~~~~~~~~~~~~~~~~~~~~~~~~~~~~~~~~~~~~~~~~~~~~~~~~~~~~~~~~}\nonumber\\
=\nu_1^2+\nu_2^2+\nu_3^2{~~~~~~~~~~~~~~~~~~~~~~~~~~~~~~~~~~~~~~~~~~~~~~~~~~~~~~~~~~~~~~~~~~~~~~~~~~~}\nonumber
\nonumber
\\+\left( \frac{\partial }{{\partial x_0 }}+i\left(
{-\nu_1x_1-\nu_2x_2-\nu_3x_3} \right)\right)^{2} +
\left(\frac{\partial }{{\partial x_1 }}+\!i\left(
{\nu_1x_0-\nu_3x_2+\nu_2x_3 } \right)\right)^{2}  \nonumber \\
+\left(\frac{\partial }{{\partial x_2 }}+
i\left({\nu_2x_0+\nu_3x_1-\nu_1x_3 } \right) \right)^{2} +
\left({\frac{\partial }{{\partial x_3 }}}+ i\left( {
\nu_3x_0-\nu_2x_1+\nu_1x_2}\right) \right)^{2}
\end{eqnarray}
with
\begin{eqnarray}
\mathcal{F}(\Delta ^{\mathbb{R}^3 })=\nu_1^2+\nu_2^2+\nu_3^2
\end{eqnarray}
and
\begin{eqnarray}
\label{E18}
\mathcal{F}(\Delta_{sub}^{N_\omega})={~~~~~~~~~~~~~~~~~~~~~~~~~~~~~~~~~~~~~~~~~~~~~~~~~~~~~~~~~}\nonumber
\\\left( \frac{\partial }{{\partial x_0 }}+\!i\left(
{-\nu_1x_1-\nu_2x_2-\nu_3x_3} \right)\right)^{\!\!2} \!\! +
\left(\frac{\partial }{{\partial x_1 }}+\!i\left(
{\nu_1x_0-\nu_3x_2+\nu_2x_3 } \right)\right)^{\!\!2}  \nonumber \\
{} +  \left(\frac{\partial }{{\partial x_2 }}+\!
i\left({\nu_2x_0+\nu_3x_1-\nu_1x_3 } \right) \right)^{\!\!2} \!\! +
\left({\frac{\partial }{{\partial x_3 }}}+\! i\left( {
\nu_3x_0-\nu_2x_1+\nu_1x_2} \right) \right)^{\!\!2}\!\!.
\end{eqnarray}
By identifying $\nu_\lambda$ with $B_\lambda, \lambda = 1,2,3$ we
have
\begin{equation}
\label{E19}
\mathcal{H}_{\vec{\nu}}=-\mathcal{F}(\Delta_{sub}^{N_\omega})
\end{equation}
which is exactly the expression $(\ref {E10})$; then we can say that
the variables $(t_1, t_2, t_3 )$ are the dual variables of the
components of the magnetic field
$\vec{B}=(B_1,B_2,B_3)=(\nu_1,\nu_2,\nu_3)$; and that the Landau
operator is equal to minus the partial Fourier transform of the
sub-laplacian associated to the
Heisenberg group $\mathbb{R}^3\times_\omega\mathbb{ R}^4$.

$\mathcal{F}(\Delta ^{\mathbb{R}^3 })$ is the intensity of the
magnetic field.
\subsection{Angular momentum in $\mathbb{R}^4$}
 The { momentum operator} of a particle in $\mathbb R^3$ with
coordinates $(x_1,x_2,x_3)$ is defined by $\hat p = -i\hbar \widehat
{\nabla}$ and its components are:
\[
\hat p_x  =  - i\hbar\frac{\partial }{{\partial x_1}}\, ,\quad \hat
p_y  = - i\hbar\frac{\partial }{{\partial x_2}}\, ,\quad \hat p_z  =
- i\hbar\frac{\partial }{{\partial x_3}} \, .
\]
The {angular momentum} operator $\widehat{\vec l}$ is defined by
$\widehat {\vec l} = \widehat {\vec x} \wedge \widehat {\vec p}$,
where $\widehat {\vec x}$
 is the position operator of the particle defined by its components $(\hat x_1,\hat x_2,\hat
x_3)$. Thus $\widehat {\vec x} \wedge \widehat {\vec p}$ has
components
\begin{eqnarray*}
\hat l_{x} & =& -i\hbar\left({{x_2}\frac{\partial }{{\partial
{x_3}}}}-{{x_3}\frac{\partial }{{\partial {x_2}}}} \right)  \\
\hat l_{y} &=& -i\hbar\left({ {x_3}\frac{\partial }{{\partial
{x_1}}}}-{ {x_1}\frac{\partial }{{\partial {x_3}}}} \right) \\
\hat l_{z} &=& -i\hbar\left({{x_1}\frac{\partial }{{\partial
{x_2}}}}- {{x_2}\frac{\partial }{{\partial {x_1}}}} \right)
\end{eqnarray*}
and these operators verify the commutation relations:
\[
\left[ {\hat l_x ,\hat l_y } \right] = i\hat l_z \,,\quad \left[
{\hat l_y ,\hat l_z } \right] = i\hat l_x \,,\quad \left[ {\hat l_z
,\hat l_x } \right] = i\hat l_y \,.
\]
In $\mathbb R^4$ the position operator is :
\[
\widehat x = x_0 e_0  + x_1 e_1  + x_2 e_2  + x_3 e_3
\]
and the nabla operator
\[
\widehat {\partial _x}  =
  \frac{\partial }{{\partial x_0 }}\,e_0
+ \frac{\partial }{{\partial x_1 }}\,e_1 + \frac{\partial
}{{\partial x_2 }}\,e_2 + \frac{\partial }{{\partial x_3 }}\,e_3 .
\]
Here $\widehat x $ and $\widehat {\partial _x}$ are meant to be
elements of the free module
\[
\mathcal{A}\left\langle {e_0 ,e_1 ,e_2 ,e_3 } \right\rangle  \cong
\mathcal{A} \otimes _{ \mathbb{R}} \mathbb{R}\left\langle {e_0 ,e_1
,e_2 ,e_3 } \right\rangle
\]
 with basis $\left( {e_0 ,e_1 ,e_2 ,e_3 } \right)$  and coefficients in the {Weyl
algebra} $\mathcal A$ of linear operators on $C^{\infty}(\mathbb R^4
; \mathbb C)$ with generators:
 the identity map on $C^{\infty}(\mathbb R^4 ; \mathbb C)$,
 multiplication by $x_\alpha$,
 partial derivatives $\frac{\partial }{{\partial x_\alpha  }}$ and relations: $ [ {\frac{\partial } {{\partial
x_\alpha  }},x_\beta  } ] = \delta _{\alpha \beta }$, for $\alpha
,\beta = 0,...,3$, the other commutators being zero.
We now consider
the exterior product $\widehat x \wedge \widehat {\partial _x}$ ,
belonging to the module $\Lambda ^2 \left( \mathcal{A}\left\langle
{e_0 ,e_1 ,e_2 ,e_3 } \right\rangle \right)$ which is contained in
$\mathcal A \otimes _{\mathbb R} \Lambda ^2(\mathbb R^4)$:
\begin{eqnarray*}
  \widehat x \wedge \widehat \partial _x & = & \left( {x_0 \frac{\partial }
{{\partial x_1 }} - x_1 \frac{\partial } {{\partial x_0 }}}
\right)e_0  \wedge e_1  + \left( {x_2 \frac{\partial } {{\partial
x_3 }} - x_3 \frac{\partial }
{{\partial x_2 }}} \right)e_2  \wedge e_3  +   \\
{ } & &
  + \left( {x_0 \frac{\partial }
{{\partial x_2 }} - x_2 \frac{\partial } {{\partial x_0 }}}
\right)e_0  \wedge e_2  + \left( {x_3 \frac{\partial } {{\partial
x_1 }} - x_1 \frac{\partial }
{{\partial x_3 }}} \right)e_3  \wedge e_1  +   \\
{ } & & + \left( {x_0 \frac{\partial } {{\partial x_3 }} - x_0
\frac{\partial } {{\partial x_3 }}} \right)e_0  \wedge e_3  + \left(
{x_1 \frac{\partial } {{\partial x_2 }} - x_2 \frac{\partial }
{{\partial x_1 }}} \right)e_1  \wedge e_2 .
\end{eqnarray*}
In the orthonormal direct basis $\left( {e_0 ,e_1 ,e_2 ,e_3 }
\right)$ of $\mathbb R^4$, the Hodge $\star$-operator is a linear
involution of $\Lambda ^2(\mathbb R^4)$ (i.e. $\star ^2 =
id_{\mathbb R^4}$) given by
\begin{eqnarray*}
  \star  \;e_0  \wedge e_1  &=& e_2  \wedge e_3  \\
  \star  \;e_0  \wedge e_2  &=& e_3  \wedge e_1  \\
  \star  \;e_0  \wedge e_3  &=& e_1  \wedge e_2.
\end{eqnarray*}
We have:\\
$ \widehat x \wedge \widehat \partial _x  +  \star \left( {\widehat
x \wedge \widehat
\partial _x } \right)  = $
 \begin{eqnarray}
 \label{momentum}
{} & = &  \left( {x_0 \frac{\partial } {{\partial x_1 }} - x_1
\frac{\partial } {{\partial x_0 }} + x_2 \frac{\partial } {{\partial
x_3 }} - x_3 \frac{\partial } {{\partial x_2 }}} \right)\left( {e_0
\wedge e_1  + e_2  \wedge e_3 } \right) + \nonumber
  \\
{} & + &   \left( {x_0 \frac{\partial } {{\partial x_2 }} - x_2
\frac{\partial } {{\partial x_0 }} + x_3 \frac{\partial } {{\partial
x_1 }} - x_1 \frac{\partial } {{\partial x_3 }}} \right)\left( {e_0
\wedge e_2  + e_3  \wedge e_1 } \right) +  \nonumber
  \\
{} & +  & \left( {x_0 \frac{\partial } {{\partial x_3 }} - x_0
\frac{\partial } {{\partial x_3 }} + x_1 \frac{\partial } {{\partial
x_2 }} - x_2 \frac{\partial } {{\partial x_1 }}} \right)\left( {e_0
\wedge e_3  + e_1  \wedge e_2 } \right).
\end{eqnarray}
The space of self-adjoint elements of  $\Lambda ^2 \left(
{\mathbb{R}^4 } \right)$ is a 3-dimensional vector space with basis
\[\left( {e_0  \wedge e_1  + e_2  \wedge e_3 ,\; e_0  \wedge e_2  + e_3  \wedge e_1 ,\;
e_0  \wedge e_3  + e_1  \wedge e_2 } \right) , \] which we identify
with the basis $\left( {e'_1 ,e'_2 ,e'_3 } \right)$ of $\mathbb
R^3$, so that the self-adjoint elements of $\Lambda ^2 \left(
\mathcal{A} \left \langle {e_0 ,e_1 ,e_2 ,e_3 } \right \rangle
\right)$ can be viewed as elements of $ \mathcal{A}\left\langle
{e'_1 ,e'_2 ,e'_3 } \right\rangle  \cong \mathcal{A} \otimes
_{\mathbb R} \mathbb{R}\left\langle {e'_1 ,e'_2 ,e'_3 }
\right\rangle $.

We define the {angular momentum operator} of a particle in $\mathbb
R^4 $ as the element $\widehat l$ of $\mathcal{A} \otimes _{\mathbb
R} \mathbb{C}\left\langle {e'_1 ,e'_2 ,e'_3 } \right\rangle$ given
by
\[ \widehat l = -i\hbar \left(\widehat x \wedge \widehat \partial _x  +  \star \;\left( {
\widehat x \wedge \widehat \partial _x } \right) \right) . \]
The components of $\widehat l$  are :
\[
\begin{gathered}
\label{momenta0}
 l_1  =  - i\hbar\left( {x_0 \frac{\partial }
{{\partial x_1 }} - x_1 \frac{\partial } {{\partial x_0 }} + x_2
\frac{\partial } {{\partial x_3 }} - x_3 \frac{\partial }
{{\partial x_2 }}} \right)  \hfill \\
  l_2  =  - i\hbar\left( {x_0 \frac{\partial }
{{\partial x_2 }} - x_2 \frac{\partial } {{\partial x_0 }} + x_3
\frac{\partial } {{\partial x_1 }} - x_1 \frac{\partial }
{{\partial x_3 }}} \right)  \hfill \\
  l_3  =  - i\hbar\left( {x_0 \frac{\partial }
{{\partial x_3 }} - x_3 \frac{\partial } {{\partial x_0 }} + x_1
\frac{\partial } {{\partial x_2 }} - x_2 \frac{\partial }
{{\partial x_1 }}} \right)  \hfill \\
\end{gathered}
\]
The components of $\widehat l$ in terms of the complex variables
$z'_1=x_0 +ix_1$ and $z'_2=x_2 +ix_3$ are :
\begin{eqnarray*}
\label{momenta1}
 l_1  &=&\hbar \bigg(z'_1\frac{\partial}{\partial z'_1}
-\bar z'_1 \frac{\partial}{{\partial \bar z'_1 }}
+z'_2\frac{\partial}{{\partial z'_2}}
-\bar z'_2 \frac{\partial}{{\partial\bar z'_2 }}\bigg)  \\
  l_2  &=&-i\hbar \bigg(\bar z'_1 \frac{\partial}{{\partial z'_2}}
-\bar z'_2 \frac{\partial}{{\partial {z'_1}}}-z'_2
\frac{\partial}{{\partial\bar z'_1 }}
+{z'_1}\frac{\partial}{{\partial{\bar z'_2}}} \bigg)  \\
  l_3  &=& \hbar\bigg(\bar z'_1 \frac{\partial}{{\partial z'_2}}
-\bar z'_2 \frac{\partial}{{\partial
{z'_1}}}-{z'_1}\frac{\partial}{{\partial{\bar z'_2 }}}+z'_2
\frac{\partial}{{\partial\bar z'_1 }}
 \bigg).
\end{eqnarray*}
The expressions $(\ref {E10})$ and $(\ref{E19})$ can be written in
terms of the angular momentum as follows
\begin{eqnarray}
\mathcal{H}_{\vec{\nu}}=-\Delta+2<\vec{\nu},\vec{l}>
+\|\vec{\nu}\|^2\|x\|_{\mathbb{R}^4}
\end{eqnarray}
\subsection{Complex form of $\mathcal{H}_{\vec{\nu}}=\mathcal{H}_{\vec{B}}$}

Since $\mathbb R^4$ is isomorphic to $\mathbb C^2$, we define the
complex coordinates on $\mathbb R^4$ by $z'_1=x_0+ix_1$ and
$z'_2=x_2+ix_3$; we have:
\begin{eqnarray*}
 \frac{\partial}{\partial
z'_1}= \frac{1}{2}(\frac{\partial}{\partial
x_0}-i\frac{\partial}{\partial x_1}),~~~~~~ \frac{\partial}{\partial
\bar z'_1}= \frac{1}{2}(\frac{\partial}{\partial
x_0}+i\frac{\partial}{\partial x_1})\\ \frac{\partial}{\partial
z'_2}= \frac{1}{2}(\frac{\partial}{\partial
x_2}-i\frac{\partial}{\partial x_3}),~~~~~~ \frac{\partial}{\partial
\bar z'_2}= \frac{1}{2}(\frac{\partial}{\partial
x_2}+i\frac{\partial}{\partial x_3})
\end{eqnarray*}
then
\begin{eqnarray*}
\label{E17}
 \mathcal{H}_{\vec{\nu}}&=&
-4\frac{\partial^2}{\partial~z'_1\partial\bar z'_1 }
-4\frac{\partial^2}{\partial z'_2\partial\bar z'_2 }
+2\nu_1\bigg(z'_1\frac{\partial}{\partial z'_1} -\bar z'_1
\frac{\partial}{{\partial \bar z'_1 }}
+z'_2\frac{\partial}{{\partial z'_2}} -\bar z'_2
\frac{\partial}{{\partial\bar z'_2 }}\bigg)\\
&& -2i\nu_2\bigg(\bar z'_1 \frac{\partial}{{\partial z'_2}} -\bar z
'_2 \frac{\partial}{{\partial {z'_1}}}-z'_2
\frac{\partial}{{\partial\bar z'_1 }}
+{z'_1}\frac{\partial}{{\partial{\bar z'_2 }}} \bigg)\\
&&+2\nu_3\bigg(\bar z'_1 \frac{\partial}{{\partial z'_2}} -\bar z'_2
\frac{\partial}{{\partial {z'_1}}}+z'_2
\frac{\partial}{{\partial\bar z'_1 }}
-{z'_1}\frac{\partial}{{\partial{\bar z'_2 }}} \bigg) +|\vec{\nu}|^2
(|z'_1|^2+|z'_2|^2).
\end{eqnarray*}
By introducing the angular momentum $\vec l$, the precedent formula
may be condensed into
\begin{eqnarray}
\label{E17'} \nonumber
 \mathcal{H}_{\vec{\nu}}&=&
-4\frac{\partial^2}{\partial~z'_1\partial\bar{z'_1}}
-4\frac{\partial^2}{\partial z'_2\partial\bar z'_2 }
 +2<\vec{\nu},\vec{l}>
 +|\vec{\nu}|^2 (|z'_1|^2+|z'_2|^2).
\end{eqnarray}
\section{Canonical Landau operator $\mathcal{H}_{\|\vec{\nu}\|}$}
\subsection{Canonical form of the matrix $\Omega_{\vec \nu}$}
 As we have seen above, for each uniform magnetic field
$\vec{B}=\vec{\nu}=(\nu_1, \nu_2, \nu_3)$, we associate the linear
operator $\Omega_{\vec{\nu}}$ acting on $\mathbb R^{4}$, defined by
its matrix in the canonical basis $(e_0, e_1, e_2,e_3)$
\begin{displaymath}
\Omega_{\vec{\nu}} = \left(\begin{array}{cccc} 0 & -\nu_1 & -\nu_2 & -\nu_3 \\
\nu_1 &\phantom- 0 & -\nu_3 &\phantom- \nu_2 \\
\nu_2 &\phantom- \nu_3 &\phantom- 0 & -\nu_1 \\
\nu_3 & -\nu_2 &\phantom- \nu_1 &\phantom- 0
\end{array}\right)
\end{displaymath}
We remark that:
\begin{equation}
\label{E 29} {\Omega_{\vec{\nu}}}^2=-{\|\vec{\nu}\|}^2
\mathbb{I}_{4}
\end{equation}
 $\mathbb{I}_{4}$ is the identity matrix on $\mathbb R^{4}$.

Since $\textbf{i}^2=\textbf{j}^2=\textbf{k} =-1$, as particular
solutions of the equation ($\ref {E 29}$) we have:
 \begin{equation}
\label{E 30} {\Omega_{\vec{\nu}}}=\pm{\|\vec{\nu}\|}\textbf{i}
\end{equation}
 \begin{equation}
\label{E 31} {\Omega_{\vec{\nu}}}=\pm{\|\vec{\nu}\|}\textbf{j}
\end{equation}
 \begin{equation}
\label{E 32} {\Omega_{\vec{\nu}}}=\pm{\|\vec{\nu}\|}\textbf{k}
\end{equation}
We have three particular orthogonal complex structures ; these
solutions have a physical interpretation, since $\Im m\mathbb{H}\sim
\mathbb{R}^3$, the equation (\ref{E 30}) means that the magnetic
field is taken along the $x_1$ axis, the equation (\ref{E 31}) means
that the magnetic field is taken along the $x_2$ axis and the
equation (\ref{E 32}) means that the magnetic field is taken along
the $x_3$ axis. We will prove that, for any $\vec{\nu}\in \mathbb
R^{3}$ there exist an orthogonal transformation which passes from
$\Omega_{\vec{\nu}}$ to
$\Omega_{\|\vec{\nu}\|}=\|\vec{\nu}\|\textbf{i}$ or to
$\Omega_{\|\vec{\nu}\|}=\|\vec{\nu}\|\textbf{j}$ or to
$\Omega_{\|\vec{\nu}\|}=\|\vec{\nu}\|\textbf{k}$ ; $
\textbf{i},\textbf{j},\textbf{k}$ are represented by the matrices
$(\ref{E i})$, $(\ref{E j})$ and $(\ref{E k})$.
\begin{thm}
\label{thm11} For every $\vec {\nu}=(\nu_1, \nu_2, \nu_3) \in
\mathbb R^{3}$, there exists $\mathcal{R}\in \mathrm{SO(4)}$ such
that
\begin{equation}
\mathcal{R} \, \Omega_{\vec{\nu}} \,
\mathcal{R^{\rm{-1}}}=\Omega_{\|\vec{\nu}\|} \
\end{equation}
\end{thm}
\begin{proof}
 We give the proof for the case
$\Omega_{\|\vec{\nu}\|}=\|\vec{\nu}\|\textbf{i}$; the proof for the
other cases is similar to this one. We denote as above $(e_0, e_1,
e_2, e_3)$ the canonical basis of oriented Euclidian space $\mathbb
R^{4}$ and we identify a linear operator in $\mathbb R^{4}$ with its
matrix in this basis.

1. If $\nu_2=\nu_3=0$ then
\begin{displaymath}
\Omega_{\vec{\nu}} = \left(
\begin{array}{cccc}
0 & -\nu_1 & 0 & 0 \\
\nu_1 & 0 & 0 &0 \\
0 & 0 & 0 & -\nu_1 \\
0 & 0 & \nu_1 & 0
\end{array}\right)
\end{displaymath}
and if $\nu_1 \geqslant 0$, we have
$\Omega_{\|\vec{\nu}\|}=\Omega_{\vec{\nu}}$
 and we can choose for instance $\mathcal{R}=\mathbb{I}_{4}$.
 If $\nu_1<0$ , we have $\Omega_{\|\vec{\nu}\|}= -\Omega_{\vec{\nu}}$
 and we can choose
\begin{displaymath}
\mathcal{R}= \left(
\begin{array}{cccc}
1 & \phantom-0  & \phantom-0 & \phantom-0 \\
0 & -1 & \phantom-0 & \phantom-0 \\
0 & \phantom-0  & \phantom-1 & \phantom-0 \\
0 & \phantom-0  & \phantom-0 & -1
\end{array}\right) \! .
\end{displaymath}
This matrix evidently belongs to $SO(4)$.

2. Suppose $\nu_2^{2}+\nu_3^{2}\neq0$ and set
$\lambda=\sqrt{\nu_2^{2}+\nu_3^{2}}$. Define
\begin{eqnarray*}
\varepsilon_0 & :\,= & e_0 \hfill \\
\varepsilon_1 & :\,= & -{\frac{1}{\|\vec{\nu}\|}}\Omega_{\vec{\nu}}~
e_0
\end{eqnarray*}
hence
\begin{eqnarray*}
\Omega_{\vec{\nu}}~\varepsilon_0  & = & -\|\vec{\nu}\|\varepsilon_1 \hfill \\
\Omega_{\vec{\nu}}~\varepsilon_1  & = &
-{\frac{1}{\|\vec{\nu}\|}}\Omega^2_{\vec{\nu}}~e_0=\|\vec{\nu}\|\varepsilon_0
.
\end{eqnarray*}
We seek now $\varepsilon_2$ as a linear combination of $e_2$ and
$e_3$, $\varepsilon_2=ae_2+be_3$, with the orthonormality
conditions:
\[ \|\varepsilon_2\|^2=1,\quad \langle
\varepsilon_1,\varepsilon_2 \rangle_{\mathbb R^{4}}=0.
\]
We get thus $a=\pm{\frac{\nu_3}{\lambda}}$ and
$b=\mp{\frac{\nu_2}{\lambda}}$, we write $\varepsilon_2$ in the
form:
$\varepsilon_2=\frac{\nu_3}{\lambda}~e_2-\frac{\nu_2}{\lambda}~\beta
e_3$. The vector $\varepsilon_3$ is defined as \[
\varepsilon_3=-\frac{1}{\|\vec{\nu}\|}\Omega_{\vec{\nu}}~\varepsilon_2
,
\]
that is $\varepsilon_3=-\frac{\lambda}{\|\vec{\nu}\|} e_1 +
\frac{\nu_1 \nu_2}{\lambda \|\vec{\nu}\|} e_2+\frac{\nu_1
\nu_3}{\lambda \|\vec{\nu}\|} e_3$. We can easily verify that
$\Omega_{\vec{\nu}}~\varepsilon_2=-\|\vec{\nu}\|\varepsilon_3$,
$\Omega_{\vec{\nu}}~\varepsilon_3=\|\vec{\nu}\|\varepsilon_3$ and
that the matrix
\begin{displaymath}
\mathcal{R}= \left(
\begin{array}{cccc}
1 & 0 & 0 & 0 \\
0 &\frac{\nu_1}{\|\vec{\nu}\|} & 0 &-\frac{\lambda}{\|\vec{\nu}\|} \\
0 &\frac{\nu_2}{\|\vec{\nu}\|} & \frac{\nu_3}{\lambda} & \frac{\nu_1
\nu_2}
{\lambda \|\vec{\nu}\|} \\
0 & \frac{\nu_3}{\|\vec{\nu}\|} & -\frac{\nu_2}{\lambda}&
\frac{\nu_1 \nu_3}{\lambda\|\vec{\nu}\|}
\end{array}\right)
\end{displaymath}
is orthogonal with determinant 1.
\end{proof}
\begin{rem}
\textnormal{ One can easily show that the centralizer
 of $\Omega_{\|\vec{\nu}\|}$ in $ SO(4)$ is ${U(2)}$.
  For a fixed $\vec{\nu} \in \mathbb R^{3}\smallsetminus
\{0\}$,
 the set of orthonormal changes of basis $ \mathcal{B} =
\left\{ {\mathcal{R} \in SO\left( 4 \right) \, |~\mathcal{R}\,\Omega
\,_{\left\| {\,\vec \nu } \right\|} \,\mathcal{R}^{ - 1}  = \Omega
_{\vec \nu } } \right\}$ can be written $ \mathcal{B} =
\mathcal{R}_0  U\left( 2 \right)$ for some $\mathcal{R}_0 \in
\mathcal{B}$. The quotient space $SO(4) / U(2)$ is the so called
{space of orthogonal complex structures on} $\mathbb R^4$ and it is
diffeomorphic to the projective line $C\mathbb P^1$ (i.e. the sphere
$S^2$). }
\end{rem}
\begin{rem}
\textnormal{ The construction of $\mathcal{R}$ can be done in the
following way. Let $(e'_1, e'_2, e'_3)$ be the canonical basis of
the Euclidian oriented space $\mathbb{R}^3$ endowed with the vector
product $\wedge$. For $ \vec \nu  = \left( {\nu _1 ,\nu _2 ,\nu _3 }
\right) \in \mathbb{R}^3$ such that $ \nu _2^2  + \nu _3^2 \ne 0$
(i.e. $\frac{1} {{\left\| {\vec \nu } \right\|}}\vec \nu  \ne e'_1
$) we set $ \lambda  = \sqrt {\nu _2^2 + \nu _3^2 }$ and we define
\begin{eqnarray}
\label{eq4}
 \varepsilon '_1 & = &
\frac{1}{\|\vec \nu \|} \, \vec \nu \, , \\
\label{eq5}
 \varepsilon '_2 & = &
\frac{1}{\|\varepsilon'_1 \wedge e_1\|} \, \varepsilon'_1 \wedge e'_1 \, , \\
\label{eq6}
 \varepsilon '_3 & = &
\varepsilon '_1  \wedge \varepsilon '_2 \, .
\end{eqnarray}
The relating matrix between the basis $ \left( {\varepsilon '_1
,\varepsilon '_2 ,\varepsilon '_3 } \right)$ and $(e'_1,e'_2,e'_3)$
 is \[
\left( {\begin{array}{ccc}
   \hfill {\frac{{\nu _1 }}
{{\left\| {\vec \nu } \right\|}}} & \hfill 0 & \hfill { -
\frac{\lambda }
{{\left\| {\vec \nu } \right\|}}} \\
   \hfill {\frac{{\nu _2 }}
{{\left\| {\vec \nu } \right\|}}} & \hfill {\frac{{\nu _3 }}
{\lambda }} & \hfill {\frac{{\nu _1 \nu _2 }}
{{\lambda \left\| {\vec \nu } \right\|}}} \\
   \hfill {\frac{{\nu _3 }}
{{\left\| {\vec \nu } \right\|}}} & \hfill { - \frac{{\nu _2 }}
{\lambda }} & \hfill {\frac{{\nu _1 \nu _3 }}
{{\lambda \left\| {\vec \nu } \right\|}}} \\
 \end{array} } \right)
\]
so that $ \left( {\varepsilon '_1 ,\varepsilon '_2 ,\varepsilon '_3
} \right)$ is easily seen to be a direct orthonormal basis of
$\mathbb R^3$. We consider the embedding $\Phi  :  \mathbb{R}^3
\hookrightarrow \mathbb{R}^4$ given by $\Phi (x_1 e'_1  + x_2 e'_2 +
x_3 e'_3)=x_0 e_0  + x_1 e_1  + x_2 e_2  + x_3 e_3$. Thus $
\mathcal{R}\,_{\vec \nu } \, :\;= \left( {e_0 ,\Phi \left(
{\varepsilon '_1 } \right),\Phi \left( {\varepsilon '_2 }
\right),\Phi \left( {\varepsilon '_3 } \right)} \right) $ is a
direct orthonormal basis of $\mathbb{R}^4$ which represents an
element of the quotient space $SO(4)/U(2)$. Denote $\sigma$ the
coset of $U(2)$ in $SO(4)/U(2)$. When restricted to the punctured
sphere $S^2 \smallsetminus \{\varepsilon'_1\}$ of $\mathbb{R}^3$
i.e. to the set $\left\{ \vec \nu \in \mathbb{R}^3 | \left\| \vec
\nu \right\| = 1, \vec \nu \ne e'_1 \right\}$, the assignment \[
\begin{array}{*{20}c}
{\mathcal{R}:} & {S^2 \smallsetminus \left\{ {\varepsilon' _1 }
\right\}} & \longrightarrow & SO(4)/U(2) \smallsetminus \{ \sigma \} \\
 {} & {\vec \nu} & \longmapsto & \mathcal{R}_{\vec \nu}
 \end{array}
\]
is one-to-one and continuous, as we can see by the formulas
$(\ref{eq4})$, $(\ref{eq5})$ and $(\ref{eq6})$. Since
$\mathcal{R}_{\varepsilon' _1}$ can be identified with $\sigma$,
$\mathcal{R}$  extends to an homeomorphism $\mathcal{R}:S^2
\longrightarrow {{{SO\left( 4 \right)} \mathord{\left/
 {\vphantom {{SO\left( 4 \right)} {U\left( 2 \right)}}} \right.
 \kern-\nulldelimiterspace} {U\left( 2 \right)}}}$. In other words:
 each direction $\vec \nu \in S^2$ defines an orthogonal complex structure on
$\mathbb{R}^4$ to which will correspond a quaternionic Landau
operator $\mathcal{H}_{\vec \nu }.$}
\end{rem}
\subsection{canonical form of the Landau operator}
Recall that by a suitable orthogonal change of basis $\mathcal
R:(e_\alpha)_{\alpha=0}^3\rightarrow(e'_\alpha)_{\alpha=0}^3$ of
$\mathbb R^4$, we have
\[\mathcal R\Omega_{\vec\nu}\mathcal R^{-1}=\left(\begin{array}{cccc}
0&-\|\vec{\nu}\|&\phantom-0&\phantom-0\\
\|\vec{\nu}\|&\phantom-0&\phantom-0&\phantom-0\\
0&\phantom-0&\phantom-0&-\|\vec{\nu}\|\\
0&\phantom-0&\phantom-\|\vec{\nu}\|&\phantom-0
\end{array}
\right)=\Omega_{\|\vec{\nu}\|}=\|\vec{\nu}\|\textbf{i}\] where
$\textbf{i}$ is the element of the base of $\Im m \mathbb{H}$
represented by ($\ref {E i}$) and
\[\|\vec{\nu}\|=\nu_1^2+\nu_2^2+\nu_3^2.\]
 Let $(y_0,y_1,y_2,y_3)$ denote the
new coordinates in the basis $(e'_\alpha)_{\alpha=0}^3$. We consider
the connection defined by
\[A'=\left(\begin{array}{cccc}
0&-\|\vec{\nu}\|&\phantom-0&\phantom-0\\
\|\vec{\nu}\|&\phantom-0&\phantom-0&\phantom-0\\
0&\phantom-0&\phantom-0&-\|\vec{\nu}\|\\
0&\phantom-0&\phantom-\|\vec{\nu}\|&\phantom-0
\end{array}
\right)\left(
\begin{array}{c}y_0\\y_1\\y_2\\y_3 \end{array}
\right)\] that is
\[A'=\left( \begin{array}{c}-\|\vec{\nu}\|y_1\\\phantom- \|\vec{\nu}\|y_0\\-\|\vec{\nu}\|y_2\\\phantom- \|\vec{\nu}\|y_3 \end{array} \right)=
\left( \begin{array}{c}A'_0\\A'_1\\A'_2\\A'_3
\end{array} \right)\] and the covariant derivative is :
\[\nabla_\alpha^{A'}=\frac{\partial}{\partial y_\alpha}+iA'_\alpha(y).\]
The corresponding Hamilton operator is
\begin{eqnarray*}
\mathcal{H}_{\|\vec{\nu}\|}\!\!\!&=&\!\!\!
-\left(\!\frac{\partial}{\partial
y_0}-i\|\vec{\nu}\|y_1\!\right)^{\!2}
\!\!\!-\left(\!\frac{\partial}{\partial
y_1}+i\|\vec{\nu}\|y_0\!\right)^{\!2}
\!\!\!-\left(\!\frac{\partial}{\partial
y_2}-i\|\vec{\nu}\|y_3\!\right)^{\!2}
\!\!\!-\left(\!\frac{\partial}{\partial y_3}+i\|\vec{\nu}\|y_2\!\right)^{\!2}\\
&=&\!\!\!-\Delta_{(y_0,y_1,y_2,y_3)}^{\mathbb R^4}
-2i\|\vec{\nu}\|\left(y_0\frac{\partial}{\partial
y_1}-y_1\frac{\partial}{\partial y_0}
+y_2\frac{\partial}{\partial y_3}-y_3\frac{\partial}{\partial y_2}\right)\\
&&+\|\vec{\nu}\|^2\bigg(y_0^2+y_1^2+y_2^2+y_3^2\bigg).
\end{eqnarray*}
One can show from Section 4 that the Landau operator
$\mathcal{H}_{\|\vec{\nu}\|}$ may be defined as the partial Fourier
Transform of the sub-Laplacian  associated to the Heisenberg group
$\mathbb{R}\times\mathbb{R}^4$;  the Lie algebra $\mathfrak{h}'$ of
this group is generated by the following vector fields:
\begin{eqnarray*}
F'_0 = \frac{\partial}{\partial y_0}-y_1\frac{\partial}{\partial
t'}\;,\qquad
F'_1 = \frac{\partial}{\partial y_1}+y_0\frac{\partial}{\partial t'}\;,\\
F'_2 = \frac{\partial}{\partial y_2}-y_3\frac{\partial}{\partial
t'}\;,\qquad F'_3 = \frac{\partial}{\partial
y_3}+y_2\frac{\partial}{\partial t'}\phantom{\;,}
\end{eqnarray*}
and
\[T' = \frac{\partial}{\partial t'}\;.\]
 The quotient of $\mathfrak{h}$
by the ideal generated by $T_2$, $T_3$ is none other than the
Heisenberg algebra $\mathfrak{h}'$ of dimension $5$.

 The partial
Fourier transform with respect to the variable $t$ is $i\vartheta$,
$\vartheta$ is identified with $\|\vec{\nu}\|$.
\subsection{Complex form of $\mathcal{H}_{\|\vec{\nu}\|}$ }
We define the new complex coordinates on $\mathbb R^4$ by
$z_1=y_0+iy_1$ and $z_2=y_2+iy_3$.
We get,
\begin{eqnarray}
\label{H 35} \mathcal{H}_{\|\vec{\nu}\|}&=&-4\left\{
\frac{\partial}{\partial z_1}\frac{\partial}{\partial \bar z_1}
+\frac{\partial}{\partial z_2}\frac{\partial}{\partial \bar z_2}
+\frac{\|\vec{\nu}\|}{2}\bigg[\left(z_1\frac{\partial}{\partial z_1}
-\bar z_1\frac{\partial}{\partial\bar z_1}\right)\right.\nonumber \\
&+&\left.\left(z_2\frac{\partial}{\partial z_2} -\bar
z_2\frac{\partial}{\partial\bar z_2}\right)\bigg]
-\left(\frac{\|\vec{\nu}\|}{2}\right)^2\bigg(|z_1|^2+|z_2|^2\bigg)
\right\}.
\end{eqnarray}
The precedent expression $(\ref{H 35})$ is much simpler to study
than the expression $(\ref{E17})$.

From the second quantization formalism $\cite{B}$, we define the
annihilation operators in the following manner:
\[a_1=\frac{\partial}{\partial \bar z_1}+\frac{\|\vec{\nu}\|}{2} z_1,\quad
a_2=\frac{\partial}{\partial \bar z_2}+\frac{\|\vec{\nu}\|}{2} z_2\]
and the corresponding creation operators are
\[a^\dag_1=-\frac{\partial}{\partial  z_1}+\frac{\|\vec{\nu}\|}{2}\bar z_1,\quad
a^\dag_2=-\frac{\partial}{\partial  z_2}+\frac{\|\vec{\nu}\|}{2}
\bar z_2.\]
 They satisfy the following commutation relations:
\begin{eqnarray*}
[a_i , a_j] &=& [a^\dag_i , a^\dag_j ] = 0 \\ {}  [a_i , a^\dag_j]
&=& \|\vec{\nu}\| \delta _{i j}
\end{eqnarray*}
 A straightforward calculation gives
\begin{eqnarray*}
a^\dag_1a_1&=&-\frac{\partial}{\partial \bar
z_1}\frac{\partial}{\partial z_1}
+\frac{\|\vec{\nu}\|}{2}\left(z_1\frac{\partial}{\partial z_1} -\bar
z_1\frac{\partial}{\partial\bar z_1}\right)
+\left(\frac{\|\vec{\nu}\|}{2}\right)^2|z_1|^2-\frac{\|\vec{\nu}\|}{2}
\end{eqnarray*}
\begin{eqnarray*}
a^\dag_2a_2&=&-\frac{\partial}{\partial \bar
z_2}\frac{\partial}{\partial z_2}
+\frac{\|\vec{\nu}\|}{2}\left(z_2\frac{\partial}{\partial z_2} -\bar
z_2\frac{\partial}{\partial\bar z_2}\right)
+\left(\frac{\|\vec{\nu}\|}{2}\right)^2|z_2|^2-\frac{\|\vec{\nu}\|}{2}
\end{eqnarray*}
 We have,
\[\mathcal{H}_{\|\vec{\nu}\|} =4\left(a^\dag_1a_1+a^\dag_2a_2\right)+4\|\vec{\nu}\|\]
which is the Hamiltonian of two superposed uncoupled harmonic oscillators.  \\
The spectral analysis of this operator is  well known $\cite{GI}$
and references there in.
\section{Symmetry group of the Landau operator}
Here we examine a group theoretical aspect of the quantum system in
a magnetic field. It is known that the translation symmetry group
becomes noncommutative when a uniform magnetic field is introduced
into the Euclidean space.
\begin{thm}
Let $\mathcal T_a$ be the operator defined by
\[(\mathcal T_af)(x)=e^{i\langle A(a),\,x\rangle} f(x+a)\]
where $\langle~,~\rangle$ denotes the scalar product in $\mathbb
R^4$.
 Then
\[\mathcal T_a\circ\mathcal{H}_{\vec\nu}=
\mathcal{H}_{\vec\nu}\circ\mathcal T_a.\]
\end{thm}
\begin{proof}
We set, for every $x\in\mathbb R^4$, $g(x)=e^{i\langle
A(a),\,x\rangle }f(x+a)$. Thus for $\alpha=0,\ldots,3$, we have the
following expressions:
\begin{eqnarray*}
\frac{\partial}{\partial x_\alpha}g(x)&=&
\frac{\partial}{\partial x_\alpha}(e^{i\langle A(a),\,x\rangle }f(x+a))\\
&=&iA_\alpha(a)e^{i\langle A(a),\,x\rangle }f(x+a)+ e^{i\langle
A(a),\,x\rangle }\frac{\partial f}{\partial x_\alpha}(x+a)
\end{eqnarray*}
\begin{eqnarray*}
\frac{\partial^2}{\partial x_\alpha^2}g(x)&=&e^{i\langle
A(a),\,a\rangle }
\Big\{-A_\alpha^2(a)f(x+a)\\
&&+2iA_\alpha(a)\frac{\partial f}{\partial x_\alpha}(x+a)
+\frac{\partial^2 f}{\partial x_\alpha^2}(x+a)\Big\}
\end{eqnarray*}
\begin{eqnarray*}
\Delta^{\mathbb{R}^4}g(x)&=&e^{i\langle A(a),\,a\rangle }
\Big\{-\sum\limits_{\alpha=0}^3A_\alpha^2(a)f(x+a)\\
&&+2i\sum\limits_{\alpha=0}^3A_\alpha(a)\frac{\partial f}{\partial
x_\alpha}(x+a) +\Delta^{\mathbb{R}^4}f(x+a)\Big\}
\end{eqnarray*}
\begin{eqnarray*}
2i\sum\limits_{\alpha=0}^3A_\alpha(x)\frac{\partial}{\partial
x_\alpha}g(x)&=& 2ie^{i\langle A(a),\,x\rangle }
\sum\limits_{\alpha=0}^3A_\alpha(x)\Big\{iA_\alpha(a)f(x+a)\\
&&+\frac{\partial f}{\partial x_\alpha}(x+a)\Big\}\\
&=&e^{i\langle A(a),\,x\rangle }
\Big\{-2\sum\limits_{\alpha=0}^3A_\alpha(x)A_\alpha(a)f(x+a)\\
&&+2i\sum\limits_{\alpha=0}^3A_\alpha(x)\frac{\partial f}{\partial
x_\alpha}(x+a)\Big\}
\end{eqnarray*}
\begin{eqnarray*}
-\|\vec\nu\|_{\mathbb{R}^3}^2\|x\|_{\mathbb{R}^4}^2g(x)&=&
-\|\vec\nu\|_{\mathbb{R}^3}^2\|x\|_{\mathbb{R}^4}^2
e^{i\langle A(a),\,x\rangle }f(x+a)\\
&=&-\Big(\sum\limits_{\alpha=0}^3A_\alpha^2(x)\Big) e^{i\langle
A(a),\,x\rangle }f(x+a)
\end{eqnarray*}
so that
\begin{eqnarray*}
(\mathcal{H}_{\vec\nu}\circ\mathcal T_a)(f)(x)&=&(\mathcal{H}_{\vec\nu}g)(x)\\
&=&e^{i\langle A(a),\,x\rangle } \Big\{\Delta^{\mathbb{R}^4}
+2i\sum\limits_{\alpha=0}^3A_\alpha(x+a)\frac{\partial}{\partial x_\alpha}\\
&&-\sum\limits_{\alpha=0}^3(A_\alpha^2(a)+2A_\alpha(x)A_\alpha(a)
+A_\alpha^2(x))\Big\}f(x+a).
\end{eqnarray*}
Since
\begin{eqnarray*}
\sum\limits_{\alpha=0}^3(A_\alpha^2(a)+2A_\alpha(x)A_\alpha(a)+A_\alpha^2(x))&=&
\sum\limits_{\alpha=0}^3A_\alpha^2(x+a)\\
&=&\|\vec\nu\|_{\mathbb{R}^3}^2\|x+a\|_{\mathbb{R}^4}^2,
\end{eqnarray*}
we get finally
\begin{eqnarray*}
(\mathcal{H}_{\vec\nu}\circ\mathcal T_a)(f)(x)&=& e^{i\langle
A(a),\,x\rangle } \Big\{\Delta^{\mathbb{R}^4}
+2i\sum\limits_{\alpha=0}^3A_\alpha(x+a)\frac{\partial}{\partial x_\alpha}\\
&&-\|\vec\nu\|_{\mathbb{R}^3}^2\|x+a\|_{\mathbb{R}^4}^2\Big\}f(x+a)\\
&=&(\mathcal T_a\circ\mathcal{H}_{\vec\nu})(f)(x).
\end{eqnarray*}
\end{proof}
\paragraph{Remark.}
For $\alpha=0,\ldots,3$, we have
\begin{eqnarray*}
\nabla_\alpha(e^{i\langle A(a),\,x\rangle }f(x+a))&=&
\Big(\frac{\partial}{\partial x_\alpha}+iA_\alpha(x)\Big)
(e^{i\langle A(a),\,x\rangle }f(x+a))\\
&=&e^{i\langle A(a),\,x\rangle }
\Big\{\frac{\partial f}{\partial x_\alpha}(x+a)+iA_\alpha(a)f(x+a)\\
&&+iA_\alpha(x)f(x+a)\Big\}\\
&=&e^{i\langle A(a),\,x\rangle } \Big\{\frac{\partial f}{\partial
x_\alpha}(x+a)+iA_\alpha(x+a)f(x+a)\Big\}\\
&=&e^{i\langle A(a),\,x\rangle } \Big\{\frac{\partial}{\partial
x_\alpha}+iA_\alpha(x+a)\Big\}f(x+a)
\end{eqnarray*}
that is
\[\nabla_\alpha\circ\mathcal T_a(f)=\mathcal T_a\circ\nabla_\alpha(f).\]
Since
$\mathcal{H}_{\vec\nu}=\sum\limits_{\alpha=0}^3\nabla_\alpha^2$, we
obtain immediately
\[\mathcal{H}_{\vec\nu}\circ\mathcal T_a=
\mathcal T_a\circ\mathcal{H}_{\vec\nu}.\] It is to be noted that
$\mathcal T_a$ is a combination of a translation in the
$a$-direction and a gauge transformation. We have unitary
transformations, the set of these transformations  form a
noncommutative group, which is the magnetic translation group; the
law of the group is:
\begin{equation*}
\mathcal T_b \mathcal T_af(x) = e^{i\langle A(a),\,b\rangle
}e^{i\langle A(b+a),\,x\rangle }f(x+b+a)
\end{equation*}
where $a\in\mathbb R^4$ and $b\in\mathbb R^4$ are fixed vectors. It
is  also to be noted that the translation in the $a$-direction and
the  one in the $b$ direction do not commute but satisfy
 \begin{eqnarray*}
\mathcal T_b^{-1} \mathcal T_a^{-1}\mathcal T_b \mathcal
T_a=\displaystyle{e^{2i\langle A(a),\,b\rangle}}.
\end{eqnarray*}
The precedent relation means that the translation symmetry group
becomes noncommutative when a uniform magnetic field is introduced
into the Euclidean space.
\section{Concluding remarks and comments}
In the preceding sections, we have derived the Landau operator from
a Lie group theoretical approach. We have shown that there exists a
canonical form of this operator which is easier to handle. An
interesting investigation involving theta functions, which describes
a charged particle moving in a uniform magnetic field on a lattice
of $\mathbb R^4$, was done in $\cite{GIHZZ}$.
Other connections with this work via the magnetic translation group
can be found in $\cite{T}$ and $\cite{B}$ and references therein.
Tanimura $\cite{T}$ studied the magnetic translation groups in a
$n$-dimensional torus and their representations. Brown $\cite{B}$
found that the translation symmetry if an electron in a lattice in a
uniform magnetic filed is noncommutative and that the quantum system
obeys a projective representation of the translation group.

\end{document}